\documentclass[runningheads]{llncs}
\usepackage[bibliography=common]{apxproof}
\usepackage{amsfonts,amsmath,amssymb,amsthm}
\usepackage{paralist}
\usepackage{hyperref}
\usepackage{graphicx}
\usepackage{subfigure}

\usepackage{todonotes}
\usepackage{xspace}
\usepackage{complexity}

\usepackage[capitalise]{cleveref}
\Crefname{observation}{Observation}{Observations}
\Crefname{algorithm}{Algorithm}{Algorithms}
\Crefname{section}{Sect.}{Sects.}
\Crefname{observation}{Observation}{Observations}
\Crefname{lemma}{Lemma}{Lemmas}
\Crefname{lemma2}{Lemma}{Lemmas}
\Crefname{theorem2}{Theorem}{Theorems}
\Crefname{claim}{Claim}{Claims}
\Crefname{claimx}{Claim}{Claims}
\Crefname{figure}{Fig.}{Figs.}
\Crefname{enumi}{Condition}{Conditions}
\Crefname{property}{Property}{Properties}
\Crefname{assumption}{Assumption}{Assumptions}

\theoremstyle{definition}

\newtheoremrep{theorem2}[theorem]{Theorem}
\newtheoremrep{lemma2}[lemma]{Lemma}
\newtheoremrep{observation2}[theorem]{Observation}
\newtheoremrep{property2}[theorem]{Property}
\newtheoremrep{defn2}[defn]{Definition}
\newtheoremrep{clm2}[clm]{Claim}

\graphicspath{{figures/}}

\newcommand{\qedclaim}{\hfill $\blacksquare$}

\usepackage{thm-restate}

\let\oldendproof\endproof
\renewcommand\endproof{~\hfill$\qed$\oldendproof}

\begin{document}

\title{New Bounds on the Local and Global Edge-length Ratio of Planar Graphs}

%\author{Anonymized submission}
%\institute{Anonymized submission}

\authorrunning{Di Giacomo et al.}

% \authorrunning{
% Lines:
% \ifnum\value{gdlastlinecounter}>225\color{red}
% \fi  \arabic{gdlastlinecounter}}
% \titlerunning{
% Lines:
% \ifnum\value{gdlastlinecounter}>225\color{red}
% \fi  \arabic{gdlastlinecounter}}

 \author{Emilio Di Giacomo\inst{1} \and 
 Walter Didimo\inst{1} \and
 Giuseppe Liotta\inst{1} \and
 Henk Meijer\inst{2}\and
 Fabrizio Montecchiani\inst{1}\and
 Stephen Wismath\inst{3}}

 \institute{Università degli Studi di Perugia, Perugia, Italy\\
 \email{\{walter.didimo,emilio.digiacomo,giuseppe.liotta,fabrizio.mentecchiani\}@unipg.it} \and
 University College Roosevelt, The Netherlands\\
 \email{henkm.personal@gmail.com} \and
 University of Lethbridge, Canada\\
 \email{wismath@uleth.ca}}

\maketitle

\begin{abstract}
The \emph{local edge-length ratio} of a planar straight-line drawing $\Gamma$ is the largest ratio between the lengths of any pair of edges of $\Gamma$ that share a common vertex. The \emph{global edge-length ratio} of $\Gamma$ is the largest ratio between the lengths of any pair of edges of $\Gamma$. The local (global) edge-length ratio of a planar graph is the infimum over all local (global) edge-length ratios of its planar straight-line drawings. We show that there exist planar graphs with $n$ vertices whose local edge-length ratio is $\Omega(\sqrt{n})$. We then show a technique to establish upper bounds on the global (and hence local) edge-length ratio of planar graphs and~apply~it to Halin graphs and to other~families~of~graphs~having~outerplanarity~two. 

%\keywords{local edge length ratio \and global edge length ratio \and level planar drawings}
\end{abstract}

\section{Introduction}\label{se:intro}
Let $G=(V,E)$ be a graph and let $\Gamma$ be a planar straight-line drawing of $G$. We denote by $|e|_{\Gamma}$ the length of the segment representing an edge $e$ in $\Gamma$. The \emph{global edge-length ratio of $\Gamma$}, denoted as $\rho_g(\Gamma)$, is the maximum ratio between the lengths of any two edges in $\Gamma$, while the \emph{local edge-length ratio of $\Gamma$}, denoted as $\rho_\ell(\Gamma)$, is the maximum ratio between the lengths of two adjacent edges. Formally, 

%\fm{$\max_{(u,v),(z,w)\in E \times E}$?}
\begin{eqnarray*}
\rho_g(\Gamma) = \max_{(u,v),(z,w)\in E} \frac{|(u,v)|_{\Gamma}}{|(z,w)|_{\Gamma}}, \quad \quad
\rho_{\ell}(\Gamma) = \max_{(u,v),(v,w)\in E} \frac{|(u,v)|_{\Gamma}}{|(v,w)|_{\Gamma}}.
\end{eqnarray*}

The \emph{global edge-length ratio of $G$}, denoted as $\rho_g(G)$, is the infimum of $\rho_g(\Gamma)$ over all planar straight-line drawings $\Gamma$ of $G$, i.e., $\rho_g(G) = \inf_{\Gamma \in \mathcal{D}(G)} \rho_g(\Gamma)$, where $\mathcal{D}(G)$ is the set of all planar straight-line drawing of $G$. Analogously, the \emph{local edge-length ratio of $G$}, denoted as $\rho_{\ell}(G)$, is the infimum of $\rho_{\ell}(\Gamma)$ over all planar straight-line drawings $\Gamma$ of $G$, i.e., $\rho_{\ell}(G) = \inf_{\Gamma \in \mathcal{D}(G)} \rho_{\ell}(\Gamma)$.

Eades and Wormald~\cite{DBLP:journals/dam/EadesW90} prove that deciding when biconnected planar graphs have edge-length ratio one is NP-hard, and Cabello et al.~\cite{DBLP:journals/jgaa/CabelloDR07} extend this result to  triconnected instances. 
These results have motivated further studies on the existence of graph families for which the local or the global edge length ratio~does~not depend on the size of the graph. Concerning lower bounds, Borrazzo and Frati~\cite{DBLP:journals/jocg/BorrazzoF20} prove that the global edge length ratio of $n$-vertex planar 3-trees is $\Omega(n)$ and Bla\v{z}ej et al.~\cite{DBLP:journals/ijcga/BlazjFL21} prove an $\Omega(\log n)$ lower bound for $n$-vertex planar 2-trees. As for upper bounds, Lazard et al.~\cite{DBLP:journals/tcs/LazardLL19} show that every outerplanar graph has global edge length ratio $2$ and that for every $\varepsilon > 0$ there exists an outerplanar graph with global edge-length ratio larger than $2 - \varepsilon$. For the $n$-vertex planar 2-trees Borrazzo and Frati show that the global edge-length ratio is in $O(n^{0.695})$, while
 Bla\v{z}ej et al.~\cite{DBLP:journals/ijcga/BlazjFL21} prove that the local edge-length ratio of planar 2-trees is at most $4$. 
 
 % Note that an upper bound to the global edge-length ratio is an upper bound to the local edge-length ratio, and that a lower bound to the local edge-length ratio is a lower bound to the global edge-length ratio. Therefore,
Note that the lower bound on the global edge-length ratio of the $n$-vertex planar 3-trees does not imply a similar lower bound for their local edge-length ratio. Also, no constant upper bound is known for the global edge-length ratio of graphs with outerplanarity larger than one. Our contribution is twofold:

\begin{itemize}
\item  We prove that the local edge-length ratio of the $n$-vertex planar 3-trees is $\Omega(\sqrt{n})$; see \cref{se:Lower}.
     
\item We  show graph families of outerplanarity two whose global edge-length ratio is upper bounded by a constant, namely the Halin graphs, the outerplanar-cycle graphs (which generalize the cycle-cycle graphs), and the outerplanar-caterpillar graphs (which generalize the cycle-caterpillar graphs); see \cref{se:Upper},  
\end{itemize}

\smallskip Our approach for the lower bound is inspired by ideas~of Borrazzo and Frati~\cite{DBLP:journals/jocg/BorrazzoF20}. The upper bounds are proved with a general approach which translates the problem of computing drawings with bounded global edge length ratio to a topological question about (a variant of) level planarity with limited edge span. We remark that Halin graphs, cycle-cycle graphs, and cycle-tree graphs are well-known subjects of study in the graph drawing literature; see, e.g.~\cite{DBLP:journals/dcg/AngeliniBBKMRS18,DBLP:conf/iisa/BekosKR16,DBLP:conf/wads/ChaplickLGLM21,DBLP:conf/isaac/LozzoDEJ17,DBLP:journals/comgeo/GiacomoLM05}.  

%Results marked with a (clickable) ``$\star$'' are proved in the appendix. 

\section{Preliminaries}\label{se:preli}

Given a graph $G$ a \emph{drawing} of $G$ maps every vertex of $G$ to a distinct point in the plane, and every edge $e$ of $G$ to a Jordan arc connecting the two points representing the end-vertices of $e$ and not passing through any other vertex. A drawing is \emph{straight-line} if its edges are straight-line segments, and it is \emph{planar} if no two edges intersect, except possibly at a common endpoint. A \emph{planar graph} is a graph that admits a planar drawing. A planar drawing subdivides the plane into topologically connected regions, called \emph{faces}. The infinite region is called the \emph{outer face}; every other face is called an \emph{inner face}. A \emph{planar embedding} of a planar graph is an equivalence class of topologically equivalent planar drawings of $G$. A planar embedding of a connected planar graph can be described by the  circular order of the edges around each vertex together with the choice of the outer face. A planar graph with a given planar embedding is a \emph{plane graph}. The cycle delimiting the outer face of a plane graph is called the \emph{outer boundary} of the graph.  

A planar drawing is \emph{outerplanar} if all its vertices belong to the outer boundary,
and \emph{$2$-outerplanar} if removing the vertices of the outer boundary we obtain an outerplanar graph. A graph is \emph{$2$-outerplanar} (\emph{outerplanar}) if it admits a $2$-outerplanar drawing (outerplanar drawing).

A \emph{$k$-span weakly level planar drawing} ({\em k-SWLP drawing}, for short) $\Gamma$ is a planar drawing whose vertices are drawn on a set of horizontal equispaced lines, called \emph{levels} and whose edges are drawn as straight-line segments each intersecting at most $k+1$ levels. To fix the ideas in the following descriptions we assume that the levels are numbered from top to bottom. An edge that intersects $k+1$ levels has \emph{span} $k$. If no edge of $\Gamma$ has span $0$, $\Gamma$ is a \emph{$k$-span level planar drawing (k-SLP drawing, for short)}. A graph is \emph{$k$-SWPL} if it admits a $k$-SWLP drawing and it is \emph{$k$-SPL} if it admits a $k$-SLP drawing. 

Given a degree-2 vertex $v$ of a graph $G$ with incident edges $(u,v)$ and $(v,w)$, to \emph{smooth $v$} means to remove $v$ and its incident edges from $G$ and replace them with the edge $(u,w)$.

\section{Lower Bound } \label{se:Lower}

We show that for every $n \geq 9$ there exists a planar $3$-tree $G$ with $n$ vertices such that $\rho_{\ell}(G) \in \Omega(\sqrt{n})$. A \emph{planar $3$-tree} is a planar graph that it is either a triangle or that can be obtained by connecting a new vertex to the vertices of a triangular face of a planar $3$-tree.  

We define a family of plane $3$-trees $G_k$ for $k \geq 1$, such that each $G_k$ has $n=2k+1$ vertices. The graph $G_1$ is a $3$-cycle $C_1$. Assume that a plane $3$-tree $G_{k-1}$ has been defined and that its outer face is a $3$-cycle $C_{k-1}$ whose vertices are denoted as $a$, $b_{k-1}$, and $c_{k-1}$; the plane $3$-tree $G_k$ is obtained by adding two vertices $b_k$ and $c_k$, and the edges $(a,b_k)$, $(a,c_k)$, $(b_{k},c_k)$, $(b_k,b_{k-1})$, $(c_k,c_{k-1})$, and $(c_k,b_{k-1})$, embedded as shown in \Cref{fi:Gk.a}.  
Notice that $G_k$ is a plane $3$-tree and it has $2k+1$ vertices. Let $\Gamma$ be a planar straight-line drawing of $G_k$ that preserves the embedding of $G_k$. For $i=1,2,\dots,k$, we denote by $\Delta_i$ the triangle that represents the cycle $C_i$ in $\Gamma$ and by $p(\Delta_i)$ its perimeter. We can assume, without loss of generality, that the length of the shortest edge over all triangles  $\Delta_i$ is $1$; if not, we can scale uniformly the drawing so to achieve this condition. The following lemma is proved by Borrazzo and Frati~\cite[pp.140-142]{DBLP:journals/jocg/BorrazzoF20}. 

\begin{lemma}[\cite{DBLP:journals/jocg/BorrazzoF20}]\label{le:geom}
Let $\Gamma$ be a planar straight-line drawing of the plane $3$-tree $G_k$ that preserves its planar embedding; then $p(\Delta_i) > p(\Delta_{i-1})+\gamma$, with $\gamma=0.3$.
\end{lemma}

We first prove a lower bound on $\rho_{\ell}(G_k)$ that holds if we consider only drawings that preserve the planar embedding of $G_k$. We then remove this restriction.

\begin{figure}[htbp]
    \centering
%    \subfigure[]{\label{fi:Gk.a}\includegraphics[width=0.3\textwidth,page=1]{figures/Gk.pdf}}
%    \hfil
%    \subfigure[]{\label{fi:Gk.b}\includegraphics[width=0.3\textwidth,page=2]{figures/Gk.pdf}}
	\includegraphics[width=0.3\textwidth,page=1]{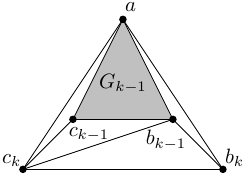}
    \caption{(a) The graph $G_k$.}
    \label{fi:Gk.a}
\end{figure}

\begin{restatable}{lemma}{leFixedEmbedding}\label{le:fixed-embedding}
Let $\Gamma$ be an embedding preserving planar straight-line drawing of the plane $3$-tree $G_k$ ; then $\rho_{\ell}(\Gamma) \geq \sqrt{\frac{k}{20}}$.
\end{restatable}
\begin{proof}
	We start by proving that $p(\Delta_i)>\gamma \cdot i$, for every $i=1,2,\dots k$. The proof is by induction on $i$. Since the shortest edge has length at least $1$, we have $p(\Delta_1)\geq 3> \gamma$. For $i > 1$, by \Cref{le:geom} we have $p(\Delta_i)>p(\Delta_{i-1})+\gamma$; by induction $p(\Delta_{i-1})>\gamma \cdot (i-1)$ and therefore $p(\Delta_{i})>\gamma \cdot i$. Thus, we have $p(\Delta_k)>\gamma \cdot k$. Let $L$ denote the length of the longest edge $e_1$ of the triangle $\Delta_k$. We have $L \geq \frac{p(\Delta_k)}{3}>\frac{\gamma \cdot k}{3}=\frac{k}{10}$. 
	Notice that, if $e_1$ is not incident to $a$, then one of the two edges of $\Delta_k$ incident to $a$ has length at least $\frac{L}{2}$. Thus, in any case we have that the longest edge $e'_1$ of $\Delta_k$ incident to $a$ (which possibly coincides with $e_1$) has length at least $\frac{k}{20}$. Consider now the shortest edge $e_2$ of $\Gamma$, whose length is $1$. This edge is either incident to $a$, or it is incident to a vertex $v$ that is adjacent to $a$ (every edge is incident to $b_i$ or to $c_i$ for some $i$, with $1 \leq i \leq k$). If $e_2$ is incident to $a$, then $\rho_{\ell}(\Gamma) \geq \frac{|e'_1|_{\Gamma}}{|e_2|_{\Gamma}} \geq  \frac{k}{20}$. If $e_2$ is not incident to $a$, then $e_3=(v,a)$ has the vertex $a$ in common with $e'_1$ and has vertex $v$ in common with $e_2$. By definition of local edge-length ratio we have $|e'_1|_{\Gamma} \leq \rho_{\ell}(\Gamma) |e_3|_{\Gamma} \leq \rho_{\ell}(\Gamma)^2 |e_2|_{\Gamma}=\rho_{\ell}(\Gamma)^2$. Thus, we have $\rho_{\ell}(\Gamma) \geq \sqrt{|e'_1|_{\Gamma}}\geq \sqrt{\frac{k}{20}}$.
\end{proof}

To prove the next theorem we  construct a planar $3$-tree with $n=4k$ vertices that in every embedding contains a copy of $G_k$ embedded as in \Cref{le:fixed-embedding}.

\begin{restatable}[{\hyperref[th:lower*]{$\star$}}]{theorem}{thlower}\label{th:lower}
For every integer $n > 9$, there exists a planar $3$-tree $G$ with $n$ vertices such that $\rho_{\ell}(G) \geq \sqrt{\frac{n}{80}}$.
\end{restatable}
\begin{proof}
	Let $G$ be a planar graph obtained as follows. Start with $K_4$; take two copies $G'$ and $G''$ of $G_k$ and glue them inside $K_4$ so that the external $3$-cycle of each copy coincides with the boundary of an internal face of $K_4$. One immediately verifies that $G$ is a planar $3$-tree. The resulting graph has $4k$ vertices and in any planar straight-line drawing of $G$, no matter what is the planar embedding of $G$, the sub-drawing of either $G'$ or $G''$ is embedded as in the proof of \Cref{le:fixed-embedding}. It follows that $\rho_{\ell}(G) \geq \sqrt{\frac{k}{20}}$. Since $n=4k$ we obtain $\rho_{\ell}(G) \geq \sqrt{\frac{n}{80}}$.
\end{proof}

\section{Upper Bounds}\label{se:Upper}

In this section we prove various upper bounds on the global edge-length ratio using a common approach based on level planar drawings.     

\begin{lemma}\label{le:general-tool}
%Let $G$ be a $k$-SLP graph for some $k\geq 1$. For any $\varepsilon > 0$ we have $\rho_{g}(G) \leq k + \varepsilon$.
If $G$ is a $k$-SLP graph for some $k\geq 1$, then $\rho_{g}(G) \leq k$.
\end{lemma}
\begin{proof}
    Let $\Gamma$ be a $k$-SLP drawing of $G$ such that the distance between two consecutive levels is one. Assume that $\Gamma$ is a straight-line drawing. If not, it is possible to achieve this condition~\cite{DBLP:journals/algorithmica/EadesFLN06}. 
    Arbitrarily chosen a value of $\varepsilon > 0$, we squeeze horizontally the drawing $\Gamma$ so that its width is $\varepsilon$. After this squeezing, for every edge $e$ we have $1 \leq |e|_{\Gamma} \leq k+\varepsilon$. It follows that $\rho_{g}(\Gamma) \leq k + \varepsilon$ and $\rho_{g}(G) \leq k$. 
\end{proof}

The next lemma follows from the fact that a $k$-SWLP drawing is easily transformed into a $(2k+1)$-SLP drawing. 

\begin{restatable}{lemma}{leGeneralToolWeakly}\label{le:general-tool-weakly}
%Let $G$ be a $k$-SWLP graph for some $k\geq 1$. For any $\varepsilon > 0$ we have $\rho_{g}(G) \leq 2k + 1 + \varepsilon$. 
If $G$ is a $k$-SWLP graph for some $k\geq 1$, then $\rho_{g}(G) \leq 2k + 1$.   
\end{restatable}
\begin{proof}
	It is easy to see that a $k$-SWLP can be transformed into a $(2k+1)$-SLP. See \Cref{fig:doubling} for an example. Namely, the edges with span zero can only connect consecutive vertices along the level. Thus, we can split each level $i$ into two levels, numbered $2i$ and $2i+1$, and assign the vertices of level $i$ alternating between $2i$ and $2i+1$. The edges with span $0$ in the original drawing will have span $1$ in the new drawing. Analogously, an edge with span $k$ in the weakly drawing will have span $2k+1$ in the proper one. Namely, suppose that for an edge $(u,v)$, $u$ is at level $i$ and $v$ is at level $j$, with $i>j$ and $i-j=k$; in the worst case $u$ is assigned to $2i+1$ and $v$ to $2j$, thus the span of $(u,v)$ in the new drawing is $2i+1-2j=2(i-j)+1=2k+1$. 
\end{proof}

\begin{figure}[htbp]
	\centering
	\subfigure[]{\label{fi:doubling-a}\includegraphics[width=0.32\textwidth,page=1]{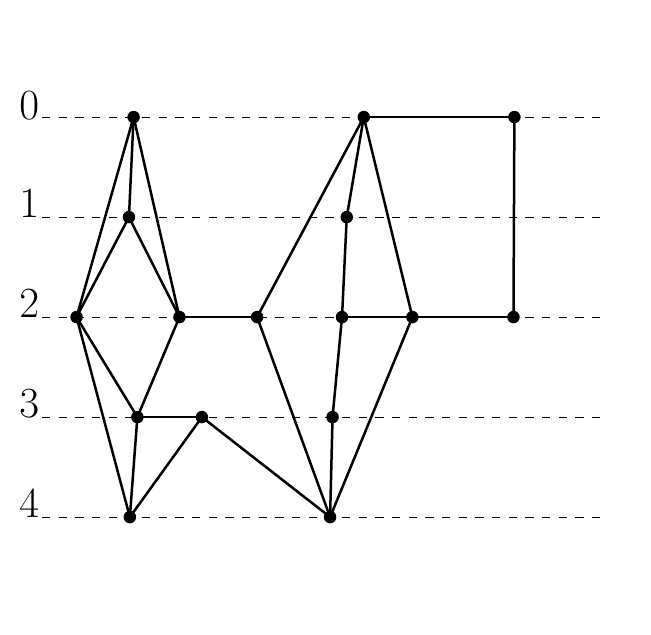}}
	\hfil
	\subfigure[]{\label{fi:doubling-b}\includegraphics[width=0.32\textwidth,page=2]{figures/doubling.pdf}}
	\caption{(a) A $2$-SWLP drawing $\Gamma$; (b) A $5$-SWLP drawing obtained from $\Gamma$ as explained in the proof of \Cref{le:general-tool-weakly}.}
	\label{fig:doubling}
\end{figure}

By \Cref{le:general-tool}, $1$-SLP graphs have global edge-length ratio at most $1$. The class of $1$-SLP graphs include bipartite outerplanar graphs, squaregraphs and dual graphs of monotone curves~\cite{DBLP:journals/algorithmica/BannisterDDEW19}. Also, by \Cref{le:general-tool-weakly} $1$-SWLP graphs, which include the outerplanar graphs~\cite{DBLP:journals/jgaa/FelsnerLW03}, have global edge-length ratio $3$. Notice however that every outerplanar graph $G$ has $\rho_{g}(G) \leq 2$~\cite{DBLP:journals/tcs/LazardLL19}.
Based on this result we look at subfamilies of $2$-outerplanar graphs. A $2$-outerplanar graph is a planar graph that has a planar embedding such that removing the vertices of the outer boundary, we obtain an outerplanar graph. We consider \emph{Halin graphs}, \emph{generalized Halin graphs}, \emph{outerplanar-cycle graphs}, and \emph{outerplanar-caterpillar graphs}, which are all sub-classes of nested pseudotrees~\cite{DBLP:conf/wads/ChaplickLGLM21} and thus have outerplanarity two. A nested pseudotree is an $2$-outerplanar graphs such that removing the vertices of the outer face we obtain a pseudotree, i.e., a connected graph with only one cycle. 
A vertex of a $2$-outerplanar graph $G$ is an \emph{outer vertex} if it belongs to the outer boundary, and an \emph{inner vertex} otherwise. An edge of $G$ is an \emph{outer edge} (\emph{inner edge}) if both its endvertices are outer (inner) vertices, and \emph{mixed edge} otherwise.

\medskip
\noindent \textbf{Halin graphs and generalized Halin graphs.} 
A \emph{Halin graph} is a $3$-connected embedded planar graph $G$ such that, by removing the edges incident to the outer face, one gets a tree whose internal vertices have degree at least $3$ and whose leaves are incident to the outerface of $G$. 
%A \emph{Halin graph} is a planar graph $G$ such that: (i) every vertex of $G$ has degree at least 3; (ii) $G$ can be decomposed into a spanning tree $T$ and a cycle $C$ through the leaves of $T$; (iii) $G$ has a planar embedding in which $C$ is the boundary of the external face.
%A \emph{Halin graph} is an embedded $2$-outerplanar graph such that removing the edges of the outer face we obtain a tree with no vertex of degree $2$. 
We call $T$ the \emph{characteristic tree} of $G$ and we call $C$ the \emph{adjoint cycle}of $G$. See \Cref{fi:halin-mistake-a,fi:halin-new-b} for two examples of two Halin graphs. In the following, we always assume that a Halin graph is an embedded graph, that is it is equipped with a planar embedding such that $C$ is the boundary of the external face. A Halin graph $G$ is \emph{trivial} if it is a wheel graph, i.e., it has only one vertex that does not belong to the boundary of the external face.  
Bannister et al.~\cite[Thm. 17]{DBLP:journals/algorithmica/BannisterDDEW19} state that all Halin graphs are $1$-SWLP which, together with \cref{le:general-tool-weakly}, would imply an upper bound of $3$ to their global edge-length ratio. Unfortunately, $K_4$ is a Halin graph that is not $1$-SWLP and the proof technique of~\cite{DBLP:journals/algorithmica/BannisterDDEW19} fails even for instances of Halin graphs different from $K_4$. Namely, let $T$ be the tree obtained by removing the edges of the outer face of the Halin graph of \cref{fi:halin-mistake-a}. The leveling  in the proof of Theorem 17 of \cite{DBLP:journals/algorithmica/BannisterDDEW19} is based on the following procedure: Choose a leaf of $T$ as the root and assign it to level $0$; at Step $i$, assign to level $i+1$ the
previously-unassigned nodes that are either children of nodes at level $i$ or that
belong to a path from such a child to its leftmost or rightmost leaf descendant in $T$.  However, for any possible choice of the root of $T$ in the Halin graph of \cref{fi:halin-mistake-a}, one obtains the leveling of \cref{fi:halin-mistake-b} that has an edge between two non-consecutive vertices on a same level.

\begin{figure}[htbp]
	\centering
	\subfigure[]{\label{fi:halin-mistake-a}\includegraphics[width=0.32\textwidth,page=1]{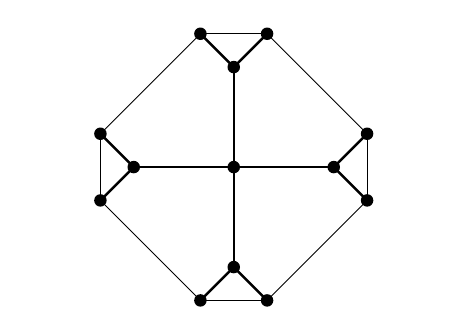}}
	\subfigure[]{\label{fi:halin-mistake-b}\includegraphics[width=0.32\textwidth,page=2]{figures/halin-mistake.pdf}}
	\caption{(a) A Halin graph; (b) a leveling obtained with the technique in~\cite{DBLP:journals/algorithmica/BannisterDDEW19}.}
	\label{fig:halin-mistake}
\end{figure}

We now prove that all Halin graphs except $K_4$ are 1-SWLP. We adopt some terminology from~\cite{DBLP:journals/comgeo/GiacomoLM05}. 

Let $T$ be an ordered rooted tree such that either $T$ consists of a single vertex or every internal vertex has at least two children. The \emph{external path} of $T$ is defined as follows. If $T$ consists of a single vertex $r$, then the external path of $T$ consists of the single vertex $r$; otherwise the external path of $T$ is the path connecting the parent $p_l$ of the leftmost leaf of $T$ and the parent $p_r$ of the rightmost leaf of $T$ (see~\Cref{fi:halin-new-a} for an illustration). Given two vertices $u$ and $v$ of an external path, we say that \emph{$u$ precedes $v$} if $u$ is encountered before $v$ when walking along the external path from $p_l$ to $p_r$.

Let $G$ be a non-trivial Halin graph and let $T$ be the characteristic tree of $G$. A \emph{tuft} of $T$ is a maximal set of at least two leaves of $T$ having the same parent, and such that this parent is adjacent to exactly one other internal vertex of $T$ (in \Cref{fi:halin-new-b} the tufts are highlighted with gray areas). Let $T^*$ be the tree obtained from $T$ by removing all leaves; we call $T^*$ the \emph{pruned tree} of $T$ (see \Cref{fi:halin-new-c} for an illustration). The following lemmas hold. 

\begin{figure}[htbp]
	\centering
	\subfigure[]{\label{fi:halin-new-a}\includegraphics[width=0.48\textwidth,page=1]{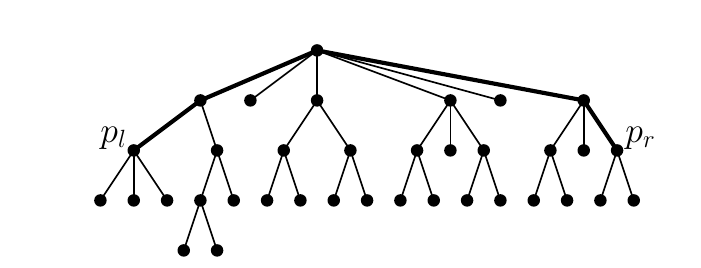}}
	\subfigure[]{\label{fi:halin-new-b}\includegraphics[width=0.48\textwidth,page=2]{figures/halin-new.pdf}}
	\subfigure[]{\label{fi:halin-new-c}\includegraphics[width=0.48\textwidth,page=3]{figures/halin-new.pdf}}
	\caption{(a) The external path of a Tree $T$; (b) A Halin graph $G$; the thick edges form the characteristic tree $T$, while the thin edges from the adjoint cycle $C$; the tufts of $T$ are highlighted with gray areas. (c) The pruned tree $T^*$ of $T$.}
	\label{fi:halin-new.1}
\end{figure}

%\todo[inline]{Questo lemma probabilmente non serve}
%
%\begin{lemma}\label{le:tuft-consecutive}
%	Let $G$ be a non-trivial Halin graph, let $T$ be the characteristic tree of $G$ and let $\tau$ be a tuft of $T$. The leaves in $\tau$ are consecutive along the adjoint cycle.
%\end{lemma}
%\begin{proof}
%	Let $v$ be the parent of the leaves in $\tau$. Suppose the leaves of $\tau$ are not consecutive. Walking clockwise along the adjoint cycle we would encounter four non-empty groups of leaves: a group $V_1$ of leaves that are in $\tau$, a group $V_2$ of leaves that are not in $\tau$ and therefore not adjacent to $v$, another group $V_3$ of leaves in $\tau$, and another group $V_4$ of leaves not in $\tau$. Since $v$ is connected both to the leaves in $V_1$ and to the leaves in $V_3$, every path connecting a leaf of $V_2$ to a leaf of $V_4$ must pass through $v$. But this would imply that $v$ has at least two neighbors that are not leaves in $T$.  
%\end{proof}

\begin{lemma}\label{le:two-tufts}
	Let $G$ be a Halin graph that is not $K_4$ and let $T$ be the characteristic tree of $G$. The number of tufts of $T$ is equal to the number of leaves of the pruned tree of $T$. 
\end{lemma}
\begin{proof}
	Let $v$ be a leaf of the pruned tree $T^*$; $v$ has degree $1$ in $T^*$ and degree at least $3$ in $T$. This means that $v$ has exactly one neighbor that is not a leaf and at least two adjacent leaves in $T$. 
	%Moreover, these leaves must be consecutive along the adjoint cycle. Suppose they are not. Walking clockwise along the adjoint cycle we would encounter four non-empty groups of leaves: a group $V_1$ of leaves whose parent is $v$, a group $V_2$ of leaves whose parent is not $v$, another group $V_3$ of leaves whose parent is $v$, and another group $V_4$ of leaves whose parent is not $v$  (see \Cref{xx}). Since $v$ is connected both to the leaves in $V_1$ and to the leaves in $V_3$, every path connecting a leaf of $V_2$ to a leaf of $V_4$ must pass through $v$. But this would imply that $v$ has degree at least two in $T^*$.  
\end{proof}

Let $G$ be a Halin graph whose characteristic tree is $T$ and whose adjoint cycle is $C$. Since we consider $G$ as an embedded graph with $C$ as the boundary of the external face, we also consider $T$ as an ordered rooted tree. In particular, let $(v_l,v_r)$ be an edge of $C$ such that $v_l$ and $v_r$ are encountered consecutively in this order when walking clockwise along $C$, and let $r$ be any vertex of the unique path connecting $v_l$ to $v_r$ in $T$. If we root $T$ at $r$, then the circular order of the edges around each vertex (derived by the planar embedding of $G$), defines a left-to-right order of the children of every vertex of $T$. According to this order, $v_l$ is the leftmost leaf of $T$ and $v_r$ is the rightmost leaf of $T$. We say that any rooting obtained in this way is a \emph{rooting defined by $(v_l,v_r)$}. Notice that if $G$ is a non-trivial Halin graph, $T$ has at least two internal vertices and therefore there is at least one edge $(v_l,v_r)$ of $C$ such that $v_l$ and $v_r$ are not siblings in $T$. When defining a rooting of the characteristic tree of a non-trivial Halin graph we always choose an edge $(v_l,v_r)$ so that $v_l$ and $v_r$ are not sibling.

We now describe a decomposition of the characteristic tree into a set of paths called \emph{characteristic paths}. The external path of $T$ is a characteristic path. For each vertex $v$ in a characteristic path $\pi$ and for each child $w$ of $v$ that is not in $\pi$, let $T'$ be the tree rooted at $w$. The external path of $T'$ is a characteristic path of $T$. Notice that each vertex $v$ of $T$ belongs to exactly one characteristic path, which we denote by $\pi(v)$. An example of the decomposition in characteristic paths is shown in~\Cref{fi:halin-new-d} for the characteristic tree of the Halin graph of~\Cref{fi:halin-new-b}. To simplify the description we define a new tree $T_{aux}$, called the \emph{auxiliary tree} of $G$, obtained from $T$ by contracting every edge of every characteristic path (see \Cref{fi:halin-new-e}). 
Let $\pi$ be a characteristic path of $T$ and let $\mu$ be the vertex of $T_{aux}$ obtained by contracting all edges of $\pi$. Vertex $\mu$ is called the \emph{image of $\pi$} in $T_{aux}$ and $\pi$ is called the \emph{pertinent path} of $\mu$ in $T$. Let $(u,v)$ be an edge of $T$ such that $\pi(u)$ and $\pi(v)$ are distinct characteristic paths; $T_{aux}$ has an edge $(\mu,\nu)$ such that $\mu$ is the image of $\pi(u)$ in $T_{aux}$ and $\nu$ is the image of $\pi(v)$ in $T_{aux}$; the edge $(\mu,\nu)$ is called the \emph{image of $(u,v)$} in $T_{aux}$ and $(u,v)$ is called the \emph{pertinent edge} of $(\mu,\nu)$ in $T$. The left-to-right order of the children of a vertex of $T_{aux}$ is the one inherited from $T$. More precisely, let $\nu_1$ and $\nu_2$ be two vertices of $T_{aux}$ with the same parent $\mu$; let $(u,v)$ be the pertinent edge of $(\mu,\nu_1)$ in $T$ and let $(w,z)$ be the pertinent edge of $(\mu,\nu_2)$ in $T$. Vertex $\nu_1$ is to the left of $\nu_2$ in $T_{aux}$ if one of the two following cases apply: (i) $u=w$ and $v$ is to the left of $z$ in the left-to-right order of the children of $u=w$ in $T$; or (ii) $u$ precedes $v$ along the characteristic path $\pi(u)=\pi(v)$. Since we only contract edges that are not incident to leaves, the leaves of $T_{aux}$ are in 1-to-1 correspondence with the leaves of $T$. Further, the leftmost child and the rightmost child of every internal vertex of $T_{aux}$ are leaves. 

\begin{figure}[htbp]
	\centering
	\subfigure[]{\label{fi:halin-new-d}\includegraphics[width=0.48\textwidth,page=4]{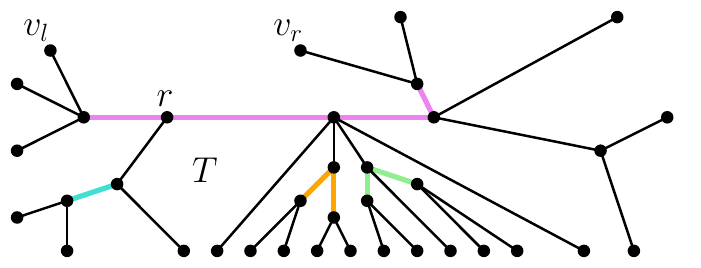}}
	\hfil
	\subfigure[]{\label{fi:halin-new-e}\includegraphics[width=0.48\textwidth,page=5]{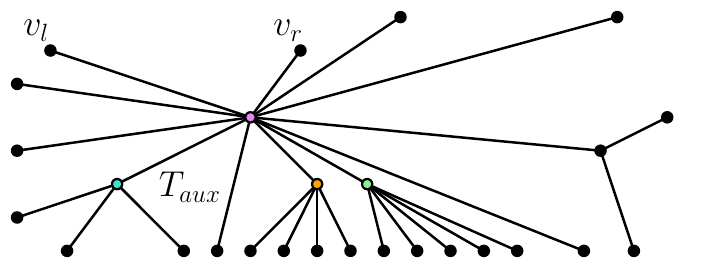}}
	%\subfigure[]{\label{fi:halin-new-f}\includegraphics[width=0.32\textwidth,page=6]{figures/halin-new.pdf}}
	\caption{(a) A decomposition in characteristic paths of the characteristic tree of the Halin graph of \Cref{fi:halin-new-b}; (b) The corresponding auxiliary tree $T_{aux}$.}
	\label{fi:halin-new.2}
\end{figure}

\begin{lemma}\label{le:characteristic-layout}
	Let $G$ be a non-trivial Halin graph and let $T$ be the characteristic tree of $G$ rooted at any non leaf vertex. Let $v_l=v_1,v_2\dots,v_k=v_r$ be the leaves of $T$ in the order they appear counterclockwise along the adjoint cycle $C$ of $G$ starting from the leftmost leaf $v_l$ of $T$ to the rightmost leaf$v_r$ of $T$. If $v_l$ belongs to a tuft and $v_r$ also belongs to a tuft, then $T$ has a 1-SWLP drawing $\Gamma$ with the following additional properties:
	\begin{itemize}
		\item[(i)]  each edge $(v_i,v_{i+1})$ can be added to $\Gamma$ without introducing any crossing;
		\item[(ii)]  the leaf $v_l$ is the first vertex of the topmost level and the leaf $v_r$ is the last vertex of the same level.
	\end{itemize}  
\end{lemma}
\begin{proof}
	We compute first a 1-SWLP drawing $\Gamma_{aux}$ of $T_{aux}$. Refer also to \Cref{fi:halin-new-f} for an illustration. For every vertex $v$ of $T_{aux}$ denote by $d(v)$ the depth of $v$ in $T_{aux}$. Assign each vertex $v$ to level $d(v)$ and order the vertices in each level by the left-to-right order of $T_{aux}$. Namely, let $u$ and $v$ be two vertices assigned to the same level $i$ and let $u'$ and $v'$ be the ancestors of $u$ and $v$, respectively, that have a common parent $w$ (possibly $u'=u$ and/or $v'=v$); $u$ precedes $v$ in level $i$ if and only if $u'$ is to the left of $v'$ in the left-to-right order of the children of $w$. The resulting drawing has span 1 and is planar by construction.
	We now replace each vertex of $T_{aux}$ by its pertinent path, thus obtaining a 1-SWLP drawing $\Gamma_{T}$ of $T$ (see \Cref{fi:halin-new-g,fi:halin-new-h1} for an illustration). More precisely, let $\mu$ be a vertex of $T_{aux}$ and let $\pi=w_1,w_2,\dots,w_h$ be its pertinent path in $T$; the vertices of $\pi$ are assigned to the same level $\ell$ as $\mu$, ordered from left to right in the order $w_1,w_2,\dots,w_h$. Clearly, the edges of $\pi$ can be drawn with span $0$ without creating any crossing. Further, no crossing is created between edges connecting vertices of $\pi$ to their subtrees. Namely, let $z_1$ be a child of $w_i$ and let $z_2$ be a child of $w_j$, for some $1 \leq i \leq j \leq h$. If $i=j$ then the two edges $(w_i,z_1)$ and $(w_j,z_2)$ do not cross. If $i < j$, then $w_i$ precedes $w_j$ along $\pi$ and therefore $w_i$ is to the left of $w_j$ along $\ell$; by the definition of $T_{aux}$, the image $\nu_1$ of $\pi(z_1)$ precedes the image $\nu_2$ of $\pi(z_2)$ in the left-to-right order of the children of $\mu$ in $T_{aux}$; as a consequence, $\nu_1$ is to the left of $\nu_2$ in level $\ell+1$ and therefore $z_1$ is to the left of $z_2$ along level $\ell+1$ once the vertices of $T_{aux}$ are replaced by their pertinent paths. 
	
	\begin{figure}[htbp]
		\centering
		\subfigure[]{\label{fi:halin-new-f}\includegraphics[width=0.48\textwidth,page=6]{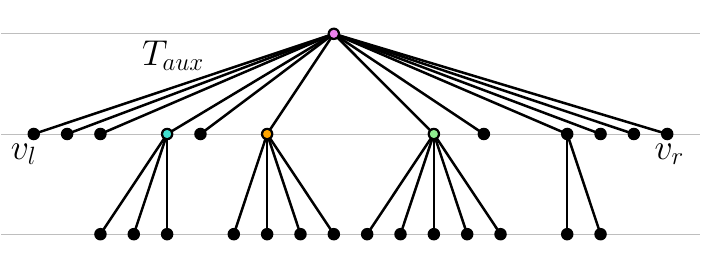}}
		\hfil
		\subfigure[]{\label{fi:halin-new-g}\includegraphics[width=0.48\textwidth,page=7]{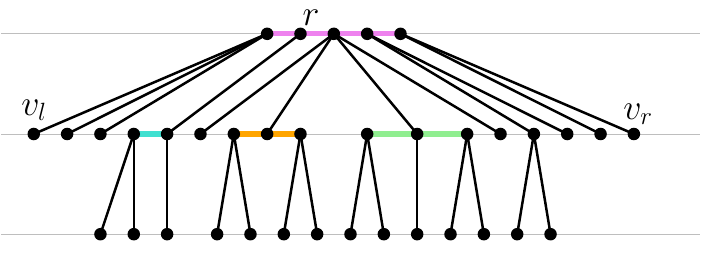}}
		\hfil
		\subfigure[]{\label{fi:halin-new-h1}\includegraphics[width=0.48\textwidth,page=8]{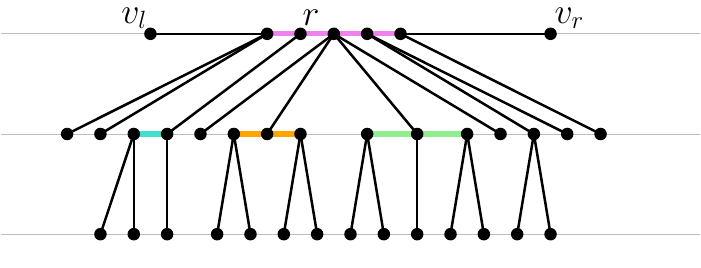}}
		\caption{(a) A 1-SWLP drawing $\Gamma_{aux}$ of $T_{aux}$; (b) A 1-SWLP drawing of $T$ obtained from $\Gamma_{aux}$ by replacing each vertex with its pertinent path. (c) A 1-SWLP drawing of $T$ with the properties of \Cref{le:characteristic-layout}.}
		\label{fi:halin-new.3}
	\end{figure}
	
	We now prove property $(i)$. Consider two leaves $v_i$ and $v_{i+1}$ that are consecutive along $C$. If they are sibling in $T_{aux}$ then they are consecutive along some level $\ell$ and therefore the edge $(v_i,v_{i+1})$ can be drawn with span $0$ without any crossing. Suppose that they are not sibling; since they are consecutive along $C$, one of the two is a leaf of a subtree whose root is a sibling of the other leaf. Suppose that $v_{i+1}$ is a leaf of a subtree whose root is a sibling of $v_i$ (the other case is analogous). In this case the root $r^*$ of the subtree containing $v_{i+1}$ is the sibling immediately to the right of $v_i$ and $v_{i+1}$ is the leftmost leaf of this subtree. Let $\ell$ be the level of $v_{i}$, we claim that $v_{i+1}$ is on level $\ell+1$. If not, the leftmost child of $r^*$ would not be a leaf, a contradiction because in $T_{aux}$ the leftmost child and the rightmost child of every internal vertex are leaves. It follows that we can add $(v_{i},v_{i+1})$ to $\Gamma$ without introducing crossings and with span $1$.   
	
	We conclude by proving property $(ii)$. Since $v_l$ belongs to a tuft, $v_2$ is a sibling of $v_l$ and therefore they are both drawn on level $1$ (i.e., the second level). Also the parent of $v_l$ and $v_2$ is the first vertex on level $0$. It follows that $v_l$ can be moved to the left of the leftmost vertex of level $0$ without creating any crossing and without increasing the maximum span. By a symmetric argument, $v_r$ can be moved to the right of the righmost vertex of level $0$ (see \Cref{fi:halin-new-h}).   
\end{proof}

Given a Halin graph $G$, if we compute a 1-SWLP drawing $\Gamma_T$ of its characteristic tree $T$ according to~\Cref{le:characteristic-layout}, we can add to $\Gamma_{T}$ all edges of the adjoint cycle except for the edge $(v_l,v_r)$ that connects the rightmost leaf to the leftmost leaf of $T$. In the next lemma we explain how to cope with this issue.

\begin{lemma}\label{le:halin}
	Every Halin graph that is not $K_4$ has a 1-SWLP drawing. 
\end{lemma}
\begin{proof}
	Let $G$ be a Halin graph that is not $K_4$, let $T$ be its characteristic tree and let $C$ be its adjoint cycle. Consider first the case when $G$ is trivial, i.e., $G$ is a wheel. Since $G$ is not $K_4$, then $T$ has at least four leaves and a 1-SWLP drawing of $G$ can be computed as shown in \Cref{xxx}. Assume therefore that $G$ is non-trivial. It follows that $T$ has at least one edge (and therefore at least four leaves). The pruned tree $T^*$ also has at least one edge and thus, by~\Cref{le:two-tufts}, $T$ has at least two tufts. We distinguish two cases depending on the existence of leaves of $T$ that do not belong to any tuft. Such leaves are called \emph{single leaves}. 
	
	\begin{description}
		\item[Case 1: $T$ has at least one single leaf.] Let $V_1$ be a maximal set of single leaves that are consecutive along $C$ and let $V_2$ be the set of vertices that are adjacent to the vertices of $V_1$ in $T$. Let $\tau_1$ and $\tau_2$ be the two tufts such that walking clockwise along $C$ $\tau_1$ is encountered before $V_1$ and $\tau_2$ is encountered after $V_1$. In this case we remove from $G$ all edges connecting the single leaves of $V_1$ to their parents; we then smooth the vertices in $V_1$, thus making a leaf $v_l$ of $\tau_1$ and a leaf $v_r$ of $\tau_2$ adjacent in $C$. Finally, we smooth the vertices in $V_2$ that have degree $2$ after the removal; notice that it may happen that no vertex of $V_2$ is removed. Let $G'$ be the resulting Halin graph and let $T'$ be the characteristic tree with a rooting defined by $(v_l,v_r)$. By~\Cref{le:characteristic-layout}, $T'$ has a 1-SWLP drawing such that all edges between consecutive leaves, except the edge $(v_l,v_r)$, can be added to the drawing without violating the planarity and without changing the maximum span. Let $\Gamma'$ be the resulting 1-SWLP drawing. Recall that $v_l$ and $v_r$ are the first and the last vertex of the topmost level, i.e., level $0$, in $\Gamma'$. We now explain how to use $\Gamma'$ to construct a 1-SWLP drawing of $G$. First we subdivide the edges that resulted from the smoothing of the vertices of $V_2$; since all these edges are drawn with span 0 in $\Gamma'$ the subdivision can be done with no crossings and no span increase. We now put the single leaves of $V_1$ on a new level $-1$ above level $0$ in the order they appear along $C$ going clockwise from $v_l$ to $v_r$. The edges connecting the leaves of $V_1$ to their parents can be added without crossings since the parents are ordered along level 0 consistently with the order of the leaves along level -1. Further, the edges connecting leaves of $V_1$ that are consecutive clockwise along $C$ can be added with span 0. Finally the edge connecting $v_r$ to the first leaf of $V_1$ and the edge connecting the last leaf of $V_1$ to $v_r$ can be added with span 1. See \Cref{fi:halin-new-i}.
		
		\begin{figure}[htbp]
			\centering
			\subfigure[]{\label{fi:halin-new-h}\includegraphics[width=0.48\textwidth,page=9]{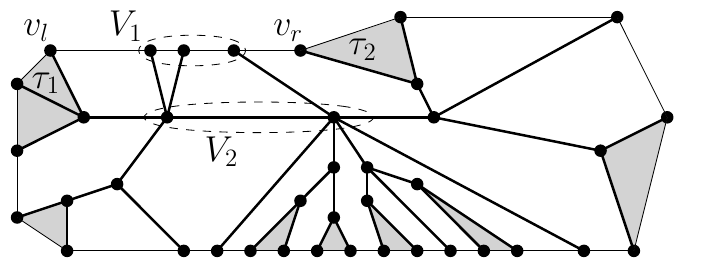}}
			\hfil
			\subfigure[]{\label{fi:halin-new-i}\includegraphics[width=0.48\textwidth,page=10]{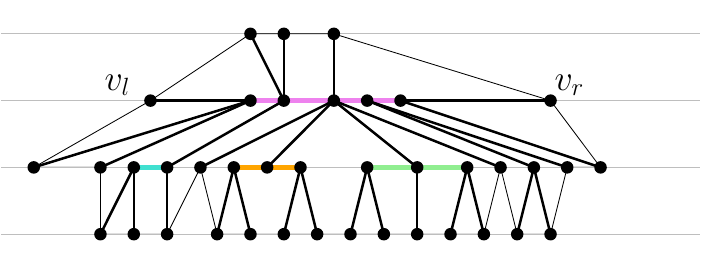}}
			\caption{Illustration of \Cref{le:halin}, Case 1: (a) A Halin graph $G$; (b) A 1-SWLP drawing of $G$.}
			\label{fi:halin-new.4}
		\end{figure}
		
		\item[Case 2: $T$ has no single leaves.]  We distinguish three sub-cases depending on the number of internal vertices of the pruned tree $T^*$ of $T$.
		\begin{description}
			\item[Case 2.a: $T^*$ has 0 internal vertices.] In this case $T^*$ is a single edge and $T$ is a caterpillar whose spine has only one edge. A 1-SWLP drawing of $G$ can be constructed as shown in \Cref{fi:halin-new-j}.
			\item[Case 2.b: $T^*$ has 1 internal vertex.] In this case $T^*$ is a star with at least three edges (otherwise $T$ would have a degree-2 vertex and $G$ would not be a Halin graph). A 1-SWLP drawing of $G$ can be constructed as shown in \Cref{fi:halin-new-k}. In the figure the internal vertex of $T^*$ has degree $4$. It is immediate to see that the construction can be extended to any degree of the internal vertex larger or equal to $3$.
			
			\item[Case 2.c: $T^*$ has at least two internal vertices.] In this case there exists one edge $e^*$ whose end-vertices are both non-leaves. This edge also exists in $T$ and if we remove it we obtain two subtrees $T_a$ and $T_b$ of $T$ (see~\Cref{fi:halin-new-l}). Let $r_a$ be the end-vertex of $e^*$ that belongs to $T_a$ and let $r_b$ be the vertex of $e^*$ that belongs to $T_b$. Rooting $T_a$ at $r_a$ and $T_b$ at $r_b$ we obtain that the leftmost leaf $v_{l,a}$ of $T_a$ is adjacent to the rightmost leaf $v_{r,b}$ of $T_b$ and the rightmost leaf $v_{r,a}$ of $T_a$ is adjacent to the leftmost leaf $v_{l,b}$ of $T_b$. Notice that $r_a$ may have degree two in $T_a$; if so we smooth it. Analogously, we smooth $r_b$ if it has degree two in $T_b$. We claim that both $T_a$ and $T_b$ have at least two tufts. Since vertex $r_a$ is not a leaf of $T^*$, it has degree at least 3 in $T^*$ and therefore it has at least two adjacent vertices in $T$. Since $r_a$ is not a leaf of $T^*$, these two adjacent vertices are not leaves of $T$ and thus they are the roots of two sub-trees, each having at least one tuft. By the same argument, $T_b$ has at least two tufts. Note that $T_a$ and $T_b$ are the characteristic trees of two Halin graphs, each obtained by connecting their leaves with a cycle. We compute a 1-SWLP drawing $\Gamma_{a}$ of $T_a$ and a 1-SWLP drawing $\Gamma_{b}$ of $T_b$ according to~\Cref{le:characteristic-layout} (see~\Cref{fi:halin-new-m} and~\Cref{fi:halin-new-n}). If $r_a$ and/or $r_b$ were smoothed we subdivide the edges obtained by the smoothing to reinsert them. We mirror the drawing $\Gamma_b$ horizontally and vertically and place it above the drawing $\Gamma_a$ (see~\Cref{fi:halin-new-o}). In this way, the first level of $\Gamma_a$ and the first level of $\Gamma_b$ (which after the vertical mirror is the last level) are consecutive. Moreover, $v_{l,a}$, $r_a$, and $v_{r,a}$ appear in this order on their level, with $v_{l,a}$ as the first vertex and $v_{r,a}$ as the last vertex; analogously $v_{r,b}$, $r_b$, and $v_{l,b}$ appear in this order on their level, with $v_{r,b}$ as the first vertex and $v_{l,b}$ as the last vertex (recall that $\Gamma_b$ has been mirrored also horizontally). This means that we can add the edges $(v_{l,a},v_{r,b})$, $(v_{r,a},v_{l,b})$, and $(r_a,r_b)$ without introducing any crossing and without increasing the maximum span.    
		\end{description}      
	\end{description}
\end{proof}

\begin{figure}[htbp]
	\centering
	\subfigure[]{\label{fi:halin-new-j}\includegraphics[width=0.16\textwidth,page=11]{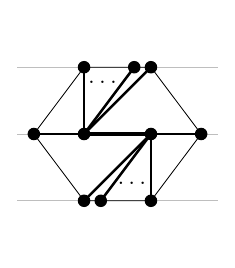}}
	\hfil
	\subfigure[]{\label{fi:halin-new-k}\includegraphics[width=0.32\textwidth,page=12]{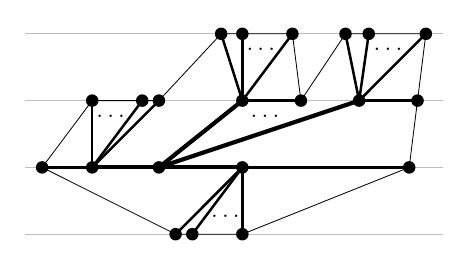}}\\
	\hfil
	\subfigure[]{\label{fi:halin-new-l}\includegraphics[width=0.48\textwidth,page=13]{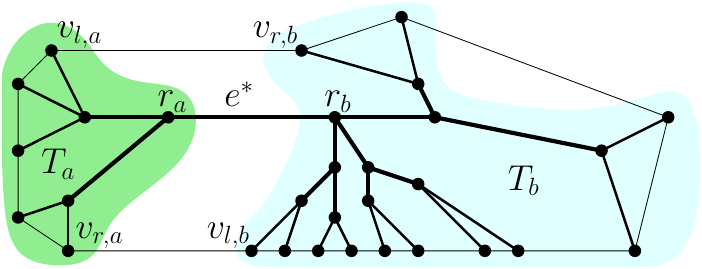}}
	\hfil
	\subfigure[]{\label{fi:halin-new-m}\includegraphics[width=0.48\textwidth,page=14]{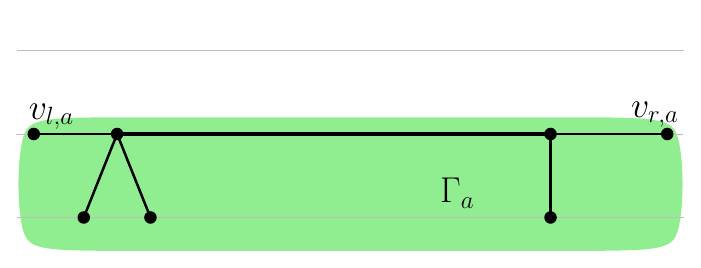}}
	\hfil
	\subfigure[]{\label{fi:halin-new-n}\includegraphics[width=0.48\textwidth,page=15]{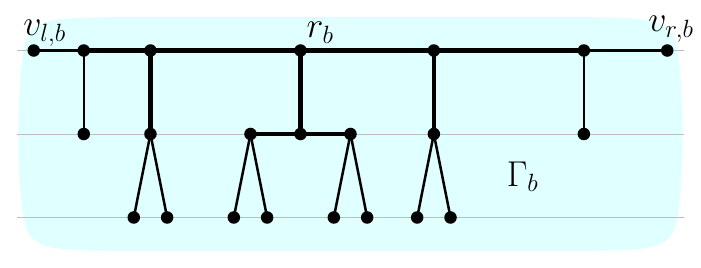}}
	\hfil
	\subfigure[]{\label{fi:halin-new-o}\includegraphics[width=0.48\textwidth,page=16]{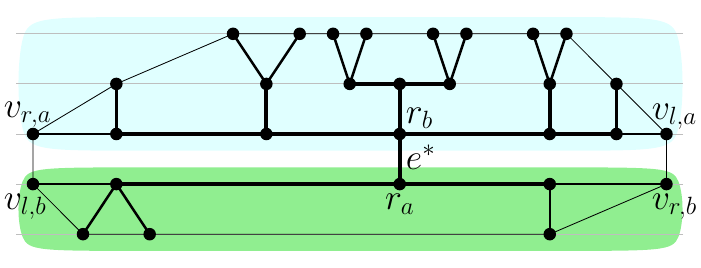}}
	\caption{Illustration of \Cref{le:halin}, (a) Case 2a; (b) Case 2b; (c)-(f) Case 2c.}
	\label{fi:halin-new.5}
\end{figure}

\Cref{le:halin} and \Cref{le:general-tool-weakly} imply the following theorem.

\begin{restatable}{theorem}{thHalin}\label{th:halin}
%Let $G$ be a Halin graph. For any $\varepsilon > 0$ we have $\rho_{g}(G) \leq 3 + \varepsilon$.  
If $G$ is a Halin graph, then $\rho_{g}(G) \leq 3$.  
\end{restatable}

A \emph{generalized Halin graph} is a planar graph obtained by subdividing the edges of a Halin graph that do not belong to the outer boundary~\cite{DBLP:conf/iisa/BekosKR16}. 

\begin{restatable}{theorem}{thGeneralizedHalin}\label{th:Generalized-Halin}
%Let $G$ be a generalized Halin graph. For any $\varepsilon > 0$ we have $\rho_{g}(G) \leq 5 + \varepsilon$.
If $G$ is a generalized Halin graph, then $\rho_{g}(G) \leq 5$.
\end{restatable}
\begin{proof}
	We prove that $G$ has a $2$-SWLP drawing $\Gamma$, which, by \Cref{le:general-tool-weakly}, implies the statement. Smoothing all the degree-2 vertices of $G$, we obtain a Halin graph $G'$. By \Cref{th:halin} $G'$ has a $2$-SWLP drawing $\Gamma'$. We now double the number of levels by placing a new level between any pair of consecutive original levels. Let $e$ be an edge of $G'$ that is obtained as the result of smoothing one or more degree-2 vertices. If $e$ has span zero in $\Gamma'$ we just subdivide it and all the degree-2 vertices are on the same level as the two endvertices of $e$. If $e$ has span one in $\Gamma'$, then it crosses at least one new level that has been added between the two levels that host its end-vertices. The intersection point of $e$ with this new level can be replaced by a chain of degree-2 vertices. Clearly, the edge that have span one in $\Gamma'$ and are not subdivided have span two in the final drawing. 
\end{proof}

\noindent \textbf{Outerplanar-cycle graphs.} An \emph{outerplanar-cycle graph} is a $2$-outerplanar graph such that the inner vertices induce a simple cycle $C$. Without loss of generality, we assume that all internal faces are triangles except the one bounded~by~$C$. If not, we can add edges to achieve this condition. We start with cycle-cycle~graphs, i.e. outerplanar-cycle graphs such that the outer vertices induce a simple cycle.

\begin{restatable}{lemma}{leCycleCyle}\label{le:cycle-cycle}
Every cycle-cycle graph has a $3$-SWLP drawing.
\end{restatable}
\begin{proof}
	Let $G$ be a cycle-cycle graph and let $u_1,u_2,\dots,u_h$ ($h \geq 3$) be the vertices of the outer cycle in clockwise order and denote by $v_1,v_2,\dots,v_k$ (with $k \geq 3$) the vertices of the inner cycle in clockwise order. Note that there must be a vertex adjacent to two inner vertices.  Assume w.l.o.g. that $u_1$ is such a vertex and that it is adjacent to vertices $v_1, v_2, \dots, v_{k'}$. See \Cref{fi:outer-cycle} for an illustration. 
	
	 \begin{figure}
	 	\centering
		\subfigure[]{\label{outer-cycle-a}\includegraphics[width=0.4\textwidth,page=1]{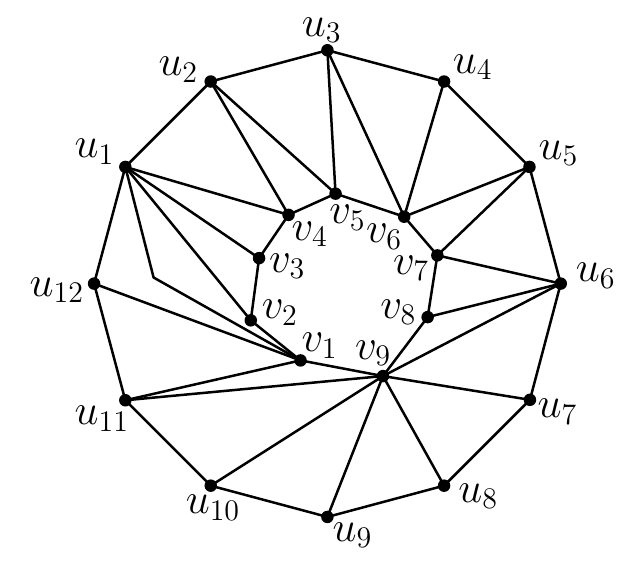}}
		\hfil
		\subfigure[]{\label{outer-cycle-b}\includegraphics[width=0.4\textwidth,page=2]{figures/outer-cycle.pdf}}
		\caption{\label{fi:outer-cycle} (a) A cycle-cycle $G$; (b) A weakly $3$-span level planar drawing of $G$.}
	\end{figure}
	
	We assign $u_1$ to level $1$, and $v_1,v_2,\dots,v_{k'}$ to level $2$ in this left-to-right order. We assign $v_{k'+1},v_{k'+2},\dots,v_k$ to level $3$ in this right-to-left order. We finally assign $u_2, u_3,\dots,u_h$ to level 4 in this right-to-left order. The edges of the inner cycle connect consecutive vertices on levels $2$ and $3$ except the edges $(v_k,v_1)$ and  $(v_{k'},v_{k'+1})$ that connect the first vertex of level $2$ to the first vertex of level $3$ and the last vertex of level $2$ to the last vertex of level $3$, respectively. The edges of the outer cycle connect consecutive vertices on level $4$ except the edges $(u_1,u_2)$ and  $(u_{h},u_{1})$, which connect the only vertex of level $1$ to the first and the last vertex of level $4$. Thus, the inner and the outer cycles are drawn planarly. The edges that connect outer vertices to inner vertices are also drawn planarly. Those connecting $u_1$ to the vertices of level $2$ are obviously planar and so are those connecting vertices on level $4$ to vertices of level $3$ (because the orderings of the vertices in the two levels are coherent). Finally, the only two edges connecting level $2$ to level $4$ are $(v_1,u_h)$ and $(v_{k'},u_2)$, which connect the first vertex of level $2$ to the first vertex level of level $4$ and the last vertex of level $2$ to the last vertex level of level $4$, respectively. the edges with span larger than $1$ are $(v_1,u_h)$ and $(v_{k'},u_2)$, which have span $2$, and $(u_1,u_h)$ and $(u_1,u_2)$, which have span $3$.   
\end{proof}

%To handle general outerplanar-cycle graphs, we use as a subroutine the technique to compute $1$-SWLP drawings of outerplanar graphs, which is based on a BFS leveling~\cite{DBLP:journals/jgaa/FelsnerLW03}. Such a drawing can be computed guaranteeing,  for a chosen edge $e=(u,v)$, either that level 1 contains only $u$ and $v$ or that $u$ and $v$ are the leftmost (or the rightmost) vertices of levels 1 and 2. 

We now consider the class of general outerplanar-cycle. We use as a subroutine the technique to compute a weakly $1$-span level planar drawing of outerplanar graphs. Such a technique is based on a BFS leveling according to which every vertex is assigned to a level equal to its distance from the vertex from which the BFS traversal starts. Also, the computed drawing preserves the outerplanar embedding of the input graph. Specifically, we can compute such a weakly level planar drawing guaranteeing that the first level of the drawing contains only a chosen edge $e=(u,v)$ (and no other vertex). To this aim, it is sufficient to add a dummy vertex $w$ connecting it to $u$ and $v$ and perform the BFS traversal starting from $w$. Alternatively, we can guarantee that for a chosen  edge $e=(u,v)$, the end-vertex $u$ is the leftmost vertex of level 1 and the end-vertex $v$ is the leftmost vertex of level 2. To this aim, it is sufficient to to start the BFS traversal from $u$ and possibly flipping the drawing in case $u$ and $v$ are the rightmost vertices of levels 1 and 2, respectively. 

    \begin{figure}[htbp]
        \centering
        \subfigure[]{\label{fi:outer-cycle-c}\includegraphics[width=0.32\textwidth,page=3]{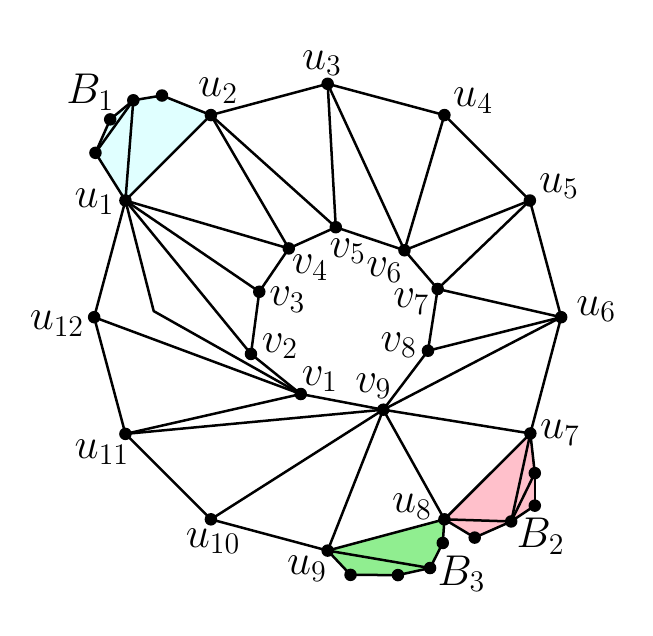}}
        \hfil
        \subfigure[]{\label{fi:outer-cycle-d}\includegraphics[width=0.32\textwidth,page=4]{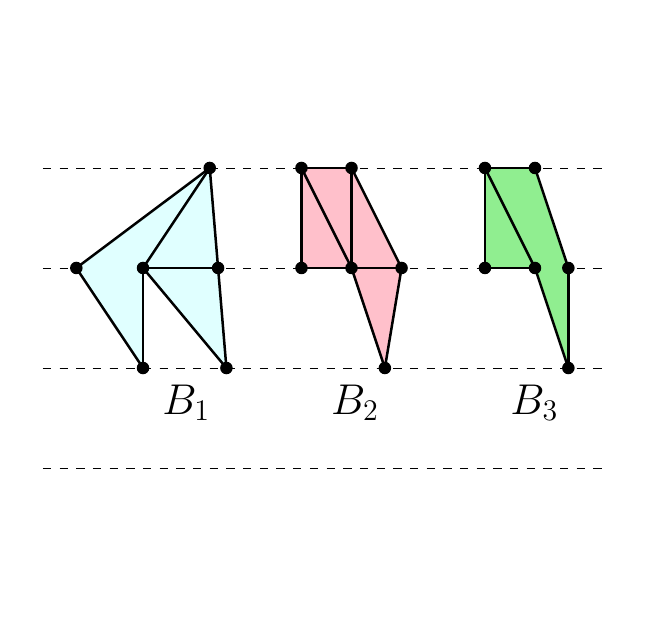}}
        \subfigure[]{\label{fi:outer-cycle-e}\includegraphics[width=0.32\textwidth,page=5]{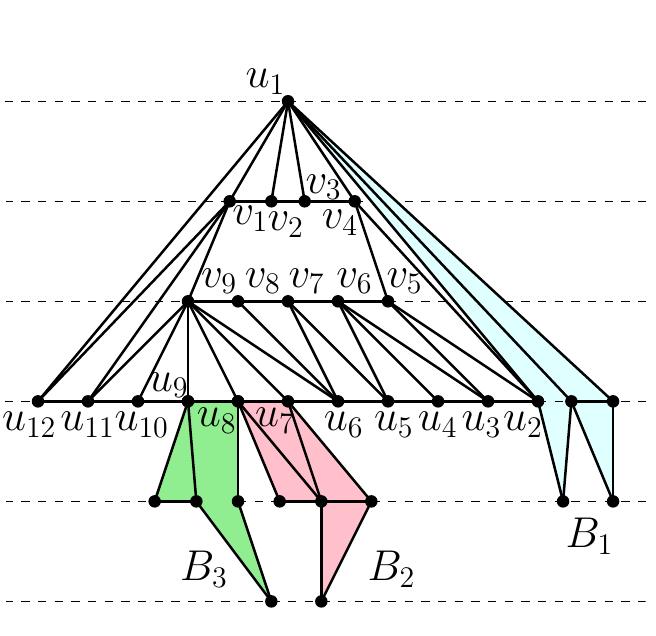}}
        \caption{\label{fi:outer-cycle-2} (a) An outerplanar-cycle $G$ with three components $B_0$, $B_1$, and $B_2$.  (b) $1$-SWLP drawings of  $B_1$, $B_2$, and $B_3$, (c) A $3$-SWLP drawing of $G$ obtained by combining the drawings of (a) and (b).}
    \end{figure}

\begin{restatable}{theorem}{thOuterplanarCycle}\label{th:outerplanar-cycle}
%Let $G$ be an outerplanar-cycle graph. For any $\varepsilon>0$ we have $\rho_{g}(G) \leq 7 + \varepsilon$. 
If $G$ is an outerplanar-cycle graph, then $\rho_{g}(G) \leq 7$. 
\end{restatable}
\begin{proof}
	Let $G$ be an outerplanar-cycle. We prove that $G$ is a $3$-SWLP which, by \Cref{le:general-tool-weakly}, implies the statement. We can decompose $G$ into a cycle-cycle $G'$ plus a set of outerplanar graphs $B_0, B_1, \dots B_h$ such that each $B_i$ shares an edge $e_i$ with the outer boundary of $G'$. See \Cref{fi:outer-cycle-2} for an illustration. We compute a $3$-SWLP drawing $\Gamma'$ as explained in the proof of \Cref{le:cycle-cycle}. We now explain how to add each component $B_i$, for $i=1,2,\dots,h$. If the edge $e_i$ shared by $G'$ and $B_i$ has span zero, i.e., it connects consecutive vertices on level 4, we compute, as explained above, a $1$-SWLP drawing $\Gamma_i$ of $B_i$ with only edge $e_i$ on level $1$. We  glue $\Gamma_i$ and $\Gamma'$ by making the two copies of $e_i$ coincident (the first level of $\Gamma_i$ coincides with the last level of $\Gamma$). Suppose now that $e_i$ is one of the two edges of $G'$ that connect the vertex $u_1$ of the first level to the vertex $u_2$ or $u_h$ of the last level; in particular assume that $e_i$ coincides with $(u_1,u_2)$ (the other case is symmetric). In this case $(u_1,u_2)$ connects the rightmost vertex of level 1 to the rightmost vertex of level 4. We compute, as explained above, a $1$-SWLP drawing $\Gamma_i$ such that the copy of $e_i$ in $\Gamma_i$ connects the leftmost vertex of level 1 to the leftmost vertex of level 2. By adding to $\Gamma_i$ two empty levels between level 1 and 2, the edge $e_i$ connects the leftmost vertex of level 1 to the leftmost vertex of level 4 in $\Gamma_i$. Hence, we can glue $\Gamma_i$ and $\Gamma'$ by making the two copies of $e_i$ coincident. The addition of the two empty levels to $\Gamma_i$ increases the span of all edges that connected vertices of the first level to the second level of $\Gamma_i$ from one to three.
\end{proof}

\noindent \textbf{Outerplanar-caterpillar graphs.} An \emph{outerplanar-caterpillar graph} is a $2$-outerplanar graph such that the inner vertices induce a caterpillar. We first handle the case of cycle-caterpillar, i.e., the case when the outer vertices induce a simple cycle. The general case can be handled similarly as in  \Cref{th:outerplanar-cycle}.  

\begin{restatable}{lemma}{leCycleCaterpillar}\label{le:cycle-caterpillar}
    Every cycle-caterpillar graph has a $4$-SWLP drawing. 
\end{restatable}
\begin{proof}
	Let $G$ be a cycle-caterpillar and denote by $u_1,u_2,\dots,u_h$ (with $h \geq 3$) the vertices of the outer cycle in clockwise order. Let $T$ be the internal caterpillar of $G$ and let $P=\langle v_1,v_2,\dots,v_k \rangle$ be the spine of $T$. 
	Similar to the case of outerplanar-cycle graphs we can assume that the graph is maximal with respect to the mixed edges, i.e., if an outer vertex and an inner vertex can be connected by an edge without creating multiple edges, we connect them.
	Let $G'$ be the graph obtained by removing all vertices that are leaves of $T$. We classify every mixed edge $e$ of $G'$ as follows. If edge $e$ is incident to $v_1$ or to $v_k$ we say that $e$ is \emph{extreme}; if $e$ is incident to the left of $P$ when walking from $v_1$ to $v_k$, we say that $e$ is a \emph{left edge}; if $e$ is incident  to the right of $P$ we say that $e$ is a \emph{right edge}. 
	
	 \begin{figure}
		     \centering
		     \subfigure[]{\label{fi:outer-cater-a}\includegraphics[width=0.32\textwidth,page=1]{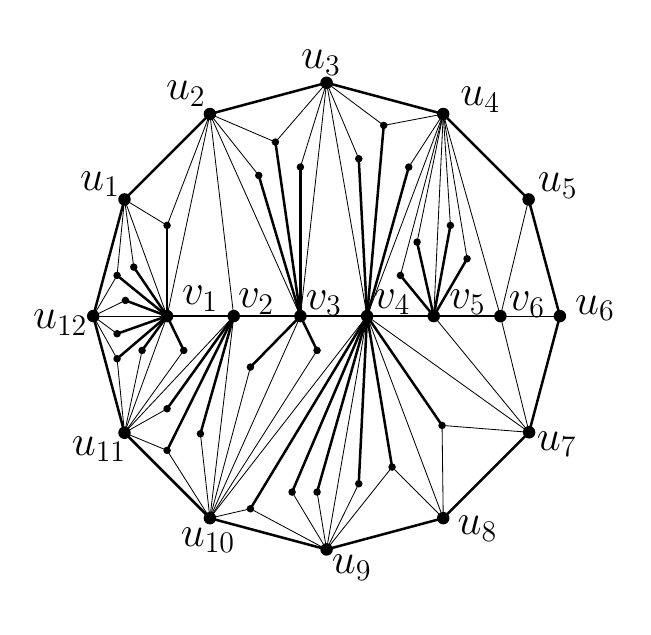}}
		     \subfigure[]{\label{fi:outer-cater-b}\includegraphics[width=0.32\textwidth,page=2]{figures/outer-cater.pdf}}
		     \subfigure[]{\label{fi:outer-cater-c}\includegraphics[width=0.32\textwidth,page=3]{figures/outer-cater.pdf}}
		     \caption{(a) A cycle-caterpillar $G$; (b) The subgraph $G'$ obtained by removing the leaves of the caterpillar; (c) A weakly $4$-span level planar drawing of $G$.}
		     \label{fig:outer-cater-1}
		 \end{figure}
	
	Assume first that no outer vertex has both a left and a right incident edge. In this case, by the maximality assumption, the internal faces of $G'$ are triangles. Further, there are two consecutive outer vertices that are both adjacent to $v_1$ and two consecutive outer vertices that are both adjacent to $v_k$. Without loss of generality assume that the two consecutive outer vertices adjacent to $v_1$ are $u_h$ and $u_1$ and let $u_i$ and $u_{i+1}$ be the two consecutive outer vertices adjacent to $v_k$. We assign vertices $u_1,u_2,\dots,u_j$ to level 1 in this order from left to right; we then assign vertices $u_{j+1},u_{j+2},\dots,u_h$ to level 5 in this order from right to left; finally we assign vertices $v_1,v_2,\dots,v_k$ to level 3 in this order from left to right. It is immediate to see that the resulting drawing is a $4$-SWLP drawing of $G'$ (actually, it is a $2$-SWLP drawing since levels $2$ and $4$ are empty). We now add the vertices of $G \setminus G'$. Each of these vertices is a leaf of $T$ and it is embedded in $G$ inside a triangular inner face of $G'$. Since every triangular inner face of $G'$ intersects level 2 or level 4 (or both), we can draw each of these vertices inside its triangle and connect it to its adjacent vertices. Note that in each triangle we can have many leaves adjacent to both end-vertices of one of the two mixed edges of the face, and at most one leaf adjacent to all the three vertices of the face (this can happen only if two of the three vertices of the face are outer vertices).
	
	\begin{figure}[htbp]
		\centering
		\subfigure[]{\label{fi:outer-cater-d}\includegraphics[width=0.32\textwidth,page=4]{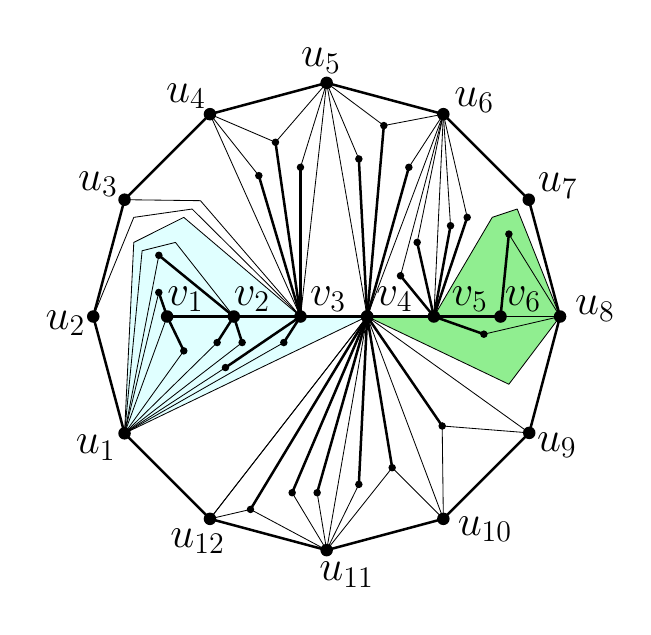}}
		\subfigure[]{\label{fi:outer-cater-e}\includegraphics[width=0.32\textwidth,page=5]{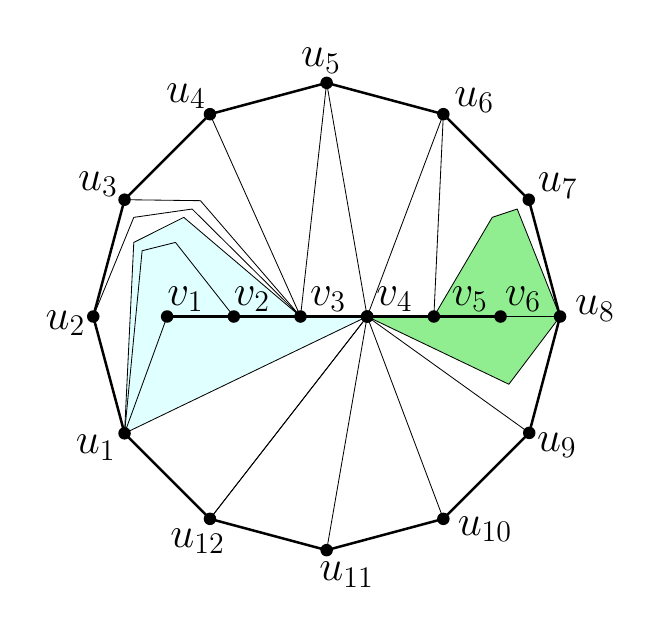}}
		\subfigure[]{\label{fi:outer-cater-f}\includegraphics[width=0.32\textwidth,page=6]{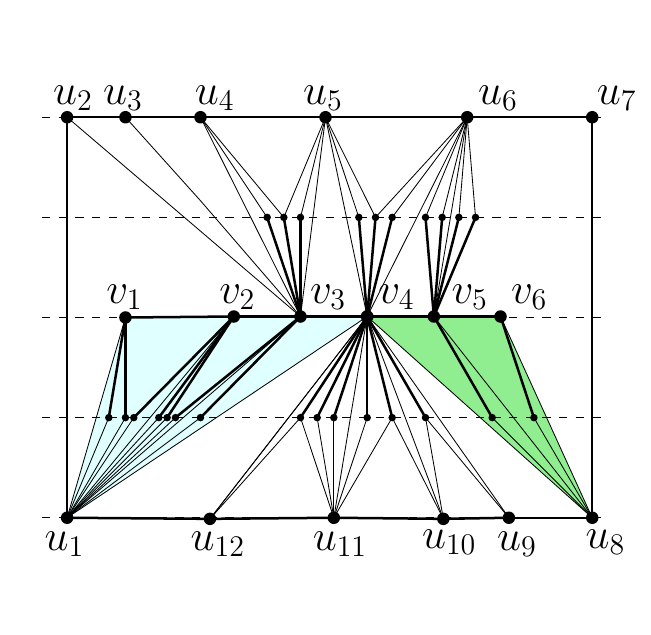}}
		\caption{(a) A cycle-caterpillar $G$; the edges of the outer vertices with both left and right inner edges are highlighted;  (b) The subgraph $G'$ obtained by removing the leaves of the caterpillar; (c) A $4$-SWLP drawing of $G$.}
		\label{fig:outer-cater-2}
	\end{figure}
	
	Suppose now that in $G'$ there are vertices that have both left and right incident edges. Notice that, there can be at most two of these vertices. One of them, say $u_1$, is adjacent to $v_1$; all the inner edges that precede $(u_1,v_1)$ in clockwise order around $u_1$ are left edges, while those that follow $(u_1,v_1)$ in clockwise order around $u_1$ are right edges. The other vertex with left and right incident edges, say $u_i$, is adjacent to $v_k$. Also in this case, the inner edges that precede $(u_i,v_k)$ around $u_i$ are right edges, while those that follow $(u_i,v_k)$ around $u_i$ are left edges. Let $v_{l_1}$ (resp. $v_{r_1}$) be the last vertex along $P$ for which there exists a left (resp. right) edge $(u_1,v_{l_1})$ (resp. $(u_1,v_{r_1})$). Analogously, Let $v_{l_2}$ (resp. $v_{r_2}$) be the first vertex along $P$ for which there exists a left (resp. right) edge $(u_i,v_{l_1})$ (resp. $(u_i,v_{r_1})$). The assignment of the vertices of $G'$ to the levels is analogous the the previous case, the only difference is that, we have to decide whether $u_1$ is assigned to level $1$ or to level $5$ and similarly for $u_i$. Notice that, any of the two choices implies to change the embedding of either the left edges or of the right edges incident to $u_1$ and similarly for $u_i$. This change of the embedding does not create any problem to draw the edges incident to $u_1$ and to $u_i$ but it can create a crossing between some of these edges and edges incident to other outer vertices. To avoid these problems we assign $u_1$ to level 1 if $v_{r_1}$ is to the right of $v_{l_1}$ along $P$ and to level $5$ otherwise. Similarly, we assign $u_i$ to level 1 if $v_{r_2}$ is to the right of $v_{l_2}$ along $P$ and to level $5$ otherwise. If $i\neq 2$, all the vertices $u_2,u_3,\dots,u_{i-1}$ are assigned to level 1; if $i \neq h$ and all vertices $u_{i+1},u_{i+2},\dots,u_h$ are assigned to level $5$. Note that it cannot be $i=2$ and $i=h$ at the same time, because the outer face of $G$ and therefore of $G'$ has at least three vertices. The vertices of $v_1,v_2,\dots,v_k$ are assigned to level 2 in this left-to-right order. It is easy to see that the resulting drawing is a $4$-SWLP of $G'$.
	
	The addition of the vertices that are leaves of $T$ is analogous to the previous case. Note that in the present case the leaves adjacent to $u_1$ or to $u_i$ may need to be drawn inside a triangular face that does not exists in the original $G'$, because we changed the embedding and moved either the left edges to the right of $P$, or the right edges to the left of $P$. However, any such leaf is only adjacent to a vertex $v_j$ of the spine and to $u_1$ (resp. to $u_i$), thus it can be drawn inside a triangular face whose boundary contains both $u_1$ (resp. $u_i$) and $v_j$.
	Finally, when there is only one outer vertex with incident mixed edges that are both left and right edges is handled similar to the cases above.
\end{proof}

\begin{restatable}[{\hyperref[th:outerplanar-caterpillar*]{$\star$}}]{theorem}{thOuterplanarCaterpillar}\label{th:outerplanar-caterpillar}
%Let $G$ be an outerplanar-caterpillar graph. For any $\varepsilon>0$ we have $\rho_{g}(G) \leq 9 + \varepsilon$. 
If $G$ is an outerplanar-caterpillar graph, then $\rho_{g}(G) \leq 9$.   
\end{restatable}
\begin{proof}
	If $G$ is a cycle-caterpillar, the result follows from \Cref{le:cycle-caterpillar}. Else, $G$ can be decomposed into a cycle-caterpillar $G'$ plus a set of outerplanar graphs $B_0, \dots, B_h$ such that each $B_i$ shares an edge $e_i$ with the outer boundary of $G'$. This case can be handled similar to the case shown in the proof of \Cref{th:outerplanar-cycle}.
\end{proof}

\section{Concluding Remarks and Open Problems}
Our lower bound holds for planar graphs with treewidth three. On the other hand, the upper bounds are for families with treewidth larger than two. Since partial 2-trees have local edge-length ratio at most four~\cite{DBLP:journals/ijcga/BlazjFL21}, we find it interesting to further study families of planar graphs with constant treewidth larger than two and constant local edge-length ratio.

% We proved an $\Omega(\sqrt{n})$ lower bound on the local edge length ratio of the $n$-vertex planar graphs having treewidth three. We also proved constant upper  bounds on the global edge-length ratio of families of graphs with outerplanarity two. These families also have treewidth larger than two. Since partial 2-trees have local edge-length ratio at most four~\cite{DBLP:journals/ijcga/BlazjFL21}, we find interesting to identify other families of planar graphs with high treewidth and constant local edge-length ratio. 

\bibliographystyle{splncs04}
\bibliography{biblio}

\begin{thebibliography}{10}
\providecommand{\url}[1]{\texttt{#1}}
\providecommand{\urlprefix}{URL }
\providecommand{\doi}[1]{https://doi.org/#1}

\bibitem{DBLP:journals/dcg/AngeliniBBKMRS18}
Angelini, P., Bruckdorfer, T., {Di Battista}, G., Kaufmann, M., Mchedlidze, T.,
  Roselli, V., Squarcella, C.: Small universal point sets for k-outerplanar
  graphs. DCG  \textbf{60}(2),  430--470 (2018).
  \doi{10.1007/s00454-018-0009-x}

\bibitem{DBLP:journals/algorithmica/BannisterDDEW19}
Bannister, M.J., Devanny, W.E., Dujmovi\'c, V., Eppstein, D., Wood, D.R.: Track
  layouts, layered path decompositions, and leveled planarity. Algorithmica
  \textbf{81}(4),  1561--1583 (2019). \doi{10.1007/s00453-018-0487-5}

\bibitem{DBLP:conf/iisa/BekosKR16}
Bekos, M.A., Kaufmann, M., Raftopoulou, C.N.: {AVDTC} of generalized 3-halin
  graphs. In: {IISA} 2016. pp.~1--6. {IEEE} (2016).
  \doi{10.1109/IISA.2016.7785421}

\bibitem{DBLP:journals/ijcga/BlazjFL21}
Bla\v{z}ej, V., Fiala, J., Liotta, G.: On edge-length ratios of partial
  2-trees. Int. J. Comput. Geom. Appl.  \textbf{31}(2-3),  141--162 (2021).
  \doi{10.1142/S0218195921500072}

\bibitem{DBLP:journals/jocg/BorrazzoF20}
Borrazzo, M., Frati, F.: On the planar edge-length ratio of planar graphs. J.
  Comput. Geom.  \textbf{11}(1),  137--155 (2020). \doi{10.20382/jocg.v11i1a6}

\bibitem{DBLP:journals/jgaa/CabelloDR07}
Cabello, S., Demaine, E.D., Rote, G.: Planar embeddings of graphs with
  specified edge lengths. J. Graph Algorithms Appl.  \textbf{11}(1),  259--276
  (2007). \doi{10.7155/jgaa.00145}

\bibitem{DBLP:conf/wads/ChaplickLGLM21}
Chaplick, S., {Da Lozzo}, G., {Di Giacomo}, E., Liotta, G., Montecchiani, F.:
  Planar drawings with few slopes of {Halin} graphs and nested pseudotrees. In:
  {WADS} 2021. Lecture Notes in Computer Science, vol. 12808, pp. 271--285.
  Springer (2021). \doi{10.1007/978-3-030-83508-8\_20}

\bibitem{DBLP:conf/isaac/LozzoDEJ17}
{Da Lozzo}, G., Devanny, W.E., Eppstein, D., Johnson, T.: Square-contact
  representations of partial 2-trees and triconnected simply-nested graphs. In:
  {ISAAC} 2017. LIPIcs, vol.~92, pp. 24:1--24:14. Schloss Dagstuhl - LZI
  (2017). \doi{10.4230/LIPIcs.ISAAC.2017.24}

\bibitem{DBLP:journals/comgeo/GiacomoLM05}
{Di Giacomo}, E., Liotta, G., Meijer, H.: Computing straight-line {3D} grid
  drawings of graphs in linear volume. Comput. Geom.  \textbf{32}(1),  26--58
  (2005). \doi{10.1016/j.comgeo.2004.11.003}

\bibitem{DBLP:journals/algorithmica/EadesFLN06}
Eades, P., Feng, Q., Lin, X., Nagamochi, H.: Straight-line drawing algorithms
  for hierarchical graphs and clustered graphs. Algorithmica  \textbf{44}(1),
  1--32 (2006). \doi{10.1007/s00453-004-1144-8}

\bibitem{DBLP:journals/dam/EadesW90}
Eades, P., Wormald, N.C.: Fixed edge-length graph drawing is {NP}-hard.
  Discret. Appl. Math.  \textbf{28}(2),  111--134 (1990).
  \doi{10.1016/0166-218X(90)90110-X}

\bibitem{DBLP:journals/jgaa/FelsnerLW03}
Felsner, S., Liotta, G., Wismath, S.K.: Straight-line drawings on restricted
  integer grids in two and three dimensions. J. Graph Algorithms Appl.
  \textbf{7}(4),  363--398 (2003). \doi{10.7155/jgaa.00075}

\bibitem{DBLP:journals/tcs/LazardLL19}
Lazard, S., Lenhart, W.J., Liotta, G.: On the edge-length ratio of outerplanar
  graphs. TCS  \textbf{770},  88--94 (2019). \doi{10.1016/j.tcs.2018.10.002}

\end{thebibliography}
\end{document}